\documentclass[conference,10pt]{IEEEtran}
\IEEEoverridecommandlockouts

\usepackage{ifthen}
\usepackage{float}
\usepackage{latexsym, amssymb, dsfont}
\usepackage{amsmath,amsthm}
\usepackage{enumerate}
\usepackage{tikz}
\usepackage{subcaption}

\newboolean{PGFPlots}
\setboolean{PGFPlots}{true}
\newboolean{SHORTversion}
\setboolean{SHORTversion}{false}

\usetikzlibrary{circuits.logic.US}
\pgfdeclarelayer{background}
\pgfsetlayers{background,main}
\usetikzlibrary{arrows,shapes,decorations,decorations.pathreplacing,decorations.pathmorphing,intersections,calc,fit}

\ifthenelse{\boolean{PGFPlots}}{
\usepackage{pgfplots}
\pgfplotsset{compat=1.6}
\usepgfplotslibrary{units}
\usetikzlibrary{external}
\tikzexternalize
\tikzset{external/only named=true}
}{}

\usepackage{textpos}
\usepackage{hyperref}

\definecolor{Spectral-4-1}{RGB}{215,25,28}
\definecolor{Spectral-4-C}{RGB}{215,25,28}
\definecolor{Spectral-4-2}{RGB}{253,174,97}
\definecolor{Spectral-4-F}{RGB}{253,174,97}
\definecolor{Spectral-4-3}{RGB}{171,221,164}
\definecolor{Spectral-4-J}{RGB}{171,221,164}
\definecolor{Spectral-4-4}{RGB}{43,131,186}
\definecolor{Spectral-4-M}{RGB}{43,131,186}

\definecolor{Spectral-PRG-1}{RGB}{202,0,32}
\definecolor{Spectral-PRG-3}{RGB}{244,165,130}
\definecolor{Spectral-PRG-2}{RGB}{186,186,186}
\definecolor{Spectral-PRG-4}{RGB}{64,64,64}

\newtheorem{theorem}{Theorem}
\newtheorem{lemma}[theorem]{Lemma}
\theoremstyle{definition}

\newtheorem{definition}[theorem]{Definition}

\gdef\dash---{\thinspace---\hskip.16667em\relax} 
\gdef\op|{\,|\;}
\newcommand{\ident}[1]{{\textit{#1}\rule{0cm}{1ex}}}

\newcommand{\figref}[1]{Fig.~\ref{#1}}

\newcommand{\bfno}[1]{\noindent{\bf #1}}

\newcommand{\IR}{\mathds{R}}

\newcommand{\dup}{\delta_\uparrow}
\newcommand{\ddo}{\delta_\downarrow}

\newcommand{\deltamin}{\delta_{\ident{min}}}

\newcommand{\Deltado}{\overline{\Delta}}

\newcommand{\etamodel}{$\boldsymbol{\eta}$-involution model}
\newcommand{\etachannel}{$\boldsymbol{\eta}$-involution channel}
\newcommand{\boeta}{\boldsymbol{\eta}}

\usepackage{graphicx}
\usepackage{caption}
\usepackage{subcaption}

\usetikzlibrary{petri}
\tikzstyle{binary place}=[place,circle, double]
\tikzstyle{node}=[circle,draw=black,thick,minimum size=9mm]
\tikzstyle{dest}=[circle,draw=black!50,fill=black!20,thick,minimum
  size=9mm]
\tikzstyle{post}=[->,thick]
\tikzstyle{pre}=[<-,thick]
\tikzstyle{every transition}=[fill,minimum width=1cm,minimum height=2mm]
\tikzstyle{Atransition}=[transition,fill,minimum width=1cm,minimum height=2mm]
\tikzstyle{Otransition}=[transition,fill=white,minimum width=1cm,minimum height=2mm]
\tikzstyle{THtransition}=[transition,fill=white,minimum width=4mm,minimum height=1cm]

\tikzstyle{Tdelay} = [draw, rectangle, rounded corners,
    minimum height=4mm, minimum width=20mm]

\tikzstyle{Tfunction} = [draw, rectangle,
    minimum height=4mm, minimum width=4mm]

\tikzstyle{Tsignal} = [draw,fill=black,circle, size=1mm]

\tikzstyle{ra} = [draw,thick,double,double distance=1.0pt,->]
\tikzstyle{r} = [draw,->,line width=0.5pt]

\def\figanfup[#1]{(1-exp(-(#1)/0.15))*(1-exp(-(#1)/0.16))}
\def\figanfdown[#1]{(1-(1-exp(-((#1)^2)/0.01))*(1-exp(-((#1)^2)/0.1)))}

\def\figchtimedist{0.35}
\def\figsigheight{0.7}
\def\figsigoffset{0.1}
\def\figchdist{1.5}

\def\tikzfigurechannelintro{
\begin{tikzpicture}[>=latex',scale=0.7,yscale=0.9,xscale=1.2]]

\coordinate (origin1) at (0,0);
\coordinate (start1) at (0,\figsigoffset);
\coordinate (origin2) at (0,-\figchdist);
\coordinate (start2) at (0,-\figchdist+\figsigoffset+\figsigheight);
\coordinate (timeorigin) at (0,{-\figchdist-\figchtimedist});
\foreach \plot/\xcapt in {1/$\text{in}(t)$, 2/$\text{out}(t)$} {
\path [draw,->]
(origin\plot) -- +(0,1.1) node (xaxis) [left,yshift=-5pt] {\small \xcapt};
\path [draw,->]
(origin\plot) -- +(6.00,0) node (yaxis) [above] {\small $t$};
}
\draw[thick] (start1) -- ++(2.5,0) -- ++(0,\figsigheight) -- ++(3.2,0);
\draw[thick] (start2) -- ++(1.3,0) -- ++(0,-\figsigheight) -- ++(3.3,0) --
++(0,\figsigheight) -- ++(1,0);

\draw[densely dashed] ($(origin1)+(2.5,0)$) -- ($(origin2)+(2.5,-3*\figsigoffset)$);

\draw[densely dashed] ($(origin2)+(1.3,0)$) -- ++(0,-3*\figsigoffset);
\draw[densely dashed] ($(origin2)+(4.6,0)$) -- ++(0,-3*\figsigoffset);

\draw[<->] ($(origin2) +(1.3,-2*\figsigoffset) $) -- ++(1.2,0) node[midway,inner sep=1pt,fill=white] {\small $T$};
\draw[<->] ($(origin2) +(2.5,-2*\figsigoffset) $) -- ++(2.1,0) node[midway,inner sep=0,fill=white] {\small $\delta(T)$};

\draw[->,densely dashed,shorten >=2pt,shorten <=2pt]
  ($(start1) + (0,0.5*\figsigheight)$) -- ($(start2) + (1.3,0)$);
\draw[->,densely dashed,shorten >=2pt,shorten <=2pt]
  ($(start1) + (2.5,\figsigheight)$) -- ($(start2) + (1.3,0) + (3.3,0)$);

\end{tikzpicture}
}

\def\figgawidth{2.8}
\def\figgaheight{1.7}
\def\figgayshift{0.15}

\def\figgamuxwidth{0.45}
\def\figgamuxheight{1}
\def\figgamuxindist{50}
\def\figgamuxtoout{0.25}
\def\figgaoutlen{0.25}
\def\figgarstlen{0.15}
\def\figgainlen{0.25}

\def\tikzfiguregate{
\begin{tikzpicture}[
gate/.style={draw, inner sep=0},
xscale=1,
yscale=1,
>=latex'
]

\coordinate (origin) at (0,0);
\coordinate (topright) at (\figgawidth,\figgaheight);

\path [draw] (origin) rectangle (topright);

\path [draw]
($(origin -| topright)!0.5!(topright) + (0,\figgayshift)$)
+(\figgaoutlen,0) node [anchor=west] {\small $v$} --
+(-\figgamuxtoout,0) node (mux) [draw, trapezium, anchor=top side,
shape border rotate=270, minimum width=\figgamuxheight*1cm, 
minimum height=\figgamuxwidth*1cm, trapezium stretches body] {};

\foreach \muxin/\muxincpt/\name/\cpt in {1/0/bool/$b$, -1/1/const/$I$} {
\path [draw] 
(mux.{180+\figgamuxindist*\muxin}) node [anchor=west, xshift=-1.5pt] {\footnotesize \muxincpt} --
++(-\csname figga\name tomux\endcsname, 0)
node (\name) [gate, anchor=east, minimum width=\csname figga\name width\endcsname*1cm,
minimum height=\csname figga\name height\endcsname*1cm] {\small \cpt};
}

\path [draw, densely dashed, thin]
(mux.north) -- (mux.north |- topright) --
++(0, \figgarstlen) node [anchor=south, inner sep=1.5pt] {\small\em reset};

\foreach \pos in {0.125, 0.375, 0.625, 0.875} {
\path [draw]
($(bool.north west)!\pos!(bool.south west)$) coordinate (tmp) --
(tmp -| origin) -- ++(-\figgainlen,0);
}
\end{tikzpicture}
}

\def\tikzfigureintroeta{
\begin{tikzpicture}[>=latex',scale=0.9]

\coordinate (origin1) at (0,0);
\coordinate (start1) at (0,\figsigoffset);
\coordinate (origin2) at (0,-\figchdist);
\coordinate (start2) at (0,-\figchdist+\figsigoffset);
\coordinate (timeorigin) at (0,{-\figchdist-\figchtimedist});
\foreach \plot/\xcapt in {1/$\text{in}(t)$, 2/$\text{out}(t)$} {
\path [draw,->]
(origin\plot) -- +(0,1.1) node (xaxis) [left,yshift=-5pt] {\small \xcapt};
\path [draw,->]
(origin\plot) -- +(5.00,0) node (yaxis) [above] {\small $t$};
}
\draw[thick] (start1) -- ++(1,0) -- ++(0,\figsigheight) -- ++(2.5,0);
\draw[thick] (start2) -- ++(3,0) -- ++(0,\figsigheight) -- ++(2,0); 
\draw[densely dashed] ($(origin1)+(0.7,\figsigheight)$) -- ++(0,-\figsigheight-3*\figsigoffset);
\draw[densely dashed] ($(origin1)+(1.4,\figsigheight)$) -- ++(0,-\figsigheight-3*\figsigoffset);
\node (eta) at ($(origin1) +(1,-2*\figsigoffset) $) {\small $\boldsymbol{\eta}$};
\draw[->] (eta) -- ($(origin1)+(0.7,-2*\figsigoffset)$);
\draw[->] (eta) -- ($(origin1)+(1.4,-2*\figsigoffset)$);
\draw[densely dashed] ($(origin2)+(1.9,\figsigheight)$) -- ++(0,-\figsigheight-4*\figsigoffset);
\draw[densely dashed] ($(origin2)+(4.1,\figsigheight)$) -- ++(0,-\figsigheight-4*\figsigoffset);
\node (deleta) at ($(origin2) +(3,-3*\figsigoffset) $) {\small $\delta(T+\boldsymbol{\eta})$};
\draw[->] (deleta) -- ($(origin2)+(1.9,-3*\figsigoffset)$);
\draw[->] (deleta) -- ($(origin2)+(4.1,-3*\figsigoffset)$);
\draw[->,densely dashed,gray,shorten >=2pt,shorten <=2pt] ($(start1) + (1,\figsigheight)$) -- ($(start2) + (3,0)$);
\end{tikzpicture}
}

\def\tikzfigureintroeta2{
\begin{tikzpicture}[>=latex',scale=0.8,yscale=0.8]

\coordinate (origin1) at (0,0);
\coordinate (start1) at (0,\figsigoffset);
\coordinate (origin2) at (0,-\figchdist);
\coordinate (start2) at (0,-\figchdist+\figsigoffset);
\coordinate (timeorigin) at (0,{-\figchdist-\figchtimedist});
\foreach \plot/\xcapt in {1/$\text{in}(t)$, 2/$\text{out}(t)$} {
\path [draw,->]
(origin\plot) -- +(0,1.1) node (xaxis) [left,yshift=-5pt] {\small \xcapt};
\path [draw,->]
(origin\plot) -- +(7.00,0) node (yaxis) [above] {\small $t$};
}
\draw[thick] (start1) -- ++(1.5,0) -- ++(0,\figsigheight) -- ++(4.6,0);
\draw[thick] ($(start2) +(0,\figsigheight)$) -- ++ (0.2,0) -- ++
(0,-\figsigheight) -- ++(3.9,0) -- ++(0,\figsigheight) -- ++(2,0);
\draw[->,densely dashed,shorten >=2pt,shorten <=2pt]
  ($(start1) + (1.5,\figsigheight)$) -- ($(start2) + (4.1,\figsigheight)$);

\draw[<->] ($(origin2) +(2.6,4.8*\figsigoffset) $) -- ++(1.5,0) node[midway,inner
sep=0, fill=white, minimum width=0.6cm] {\small $\boldsymbol{\eta^-}$};
\draw[<->] ($(origin2) +(4.1,4.8*\figsigoffset) $) -- ++(1.5,0) node[midway,inner
sep=0, fill=white, minimum width=0.6cm] {\small $\boldsymbol{\eta^+}$};

\draw[dotted, thick] ($(origin2)+(2.6,\figsigoffset)$) -- ++(0,\figsigheight) -- ++(1.5,0);
\draw[dotted, thick] ($(origin2)+(4.1,\figsigoffset)$) -- ++(1.5,0) -- ++(0,\figsigheight);
\draw[densely dashed] ($(origin2)+(0.2,\figsigoffset)$) --
($(origin2)+(0.2,-4*\figsigoffset)$); 
\draw[densely dashed] ($(origin2)+(4.1,\figsigoffset)$) --
($(origin2)+(4.1,-4*\figsigoffset)$); 
\draw[densely dashed] ($(origin1)+(1.5,\figsigoffset)$) --
($(origin2)+(1.5,-4*\figsigoffset)$);
\draw[<->] ($(origin2) +(0.2,-3*\figsigoffset) $) -- ++(1.3,0) node[midway,inner sep=0,fill=white] {\small $T$};
\draw[<->] ($(origin2) +(1.5,-3*\figsigoffset) $) -- ++(2.6,0) node[midway,inner sep=0,fill=white] {\small $\delta(T)$};
\end{tikzpicture}
}

\title{A Faithful Binary Circuit Model with\\Adversarial Noise}

\author{
  \IEEEauthorblockN{
    Matthias F{\"u}gger\IEEEauthorrefmark{1},
    J\"urgen Maier\IEEEauthorrefmark{2}
    \begin{minipage}[c]{1em}
      \href{https://orcid.org/0000-0002-0965-5746}{{\includegraphics[width=1em]{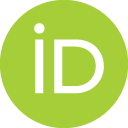}}}
    \end{minipage}\,,
    Robert Najvirt\IEEEauthorrefmark{2},
    Thomas Nowak\IEEEauthorrefmark{3}
    \begin{minipage}[c]{1em}
      \href{https://orcid.org/0000-0003-1690-9342}{{\includegraphics[width=1em]{orcID.png}}}
    \end{minipage}\,,
    Ulrich Schmid\IEEEauthorrefmark{2} 
    \begin{minipage}[c]{1em}
      \href{https://orcid.org/0000-0001-9831-8583}{{\includegraphics[width=1em]{orcID.png}}}
    \end{minipage}
  }
  \IEEEauthorblockA{\IEEEauthorrefmark{1}CNRS \& LSV, ENS Paris-Saclay}
  \IEEEauthorblockA{\IEEEauthorrefmark{2}Technische Universit\"at Wien}
  \IEEEauthorblockA{\IEEEauthorrefmark{3}Universit\'e Paris-Sud}
  \thanks{This research was supported by the
  FATAL (grant P21694) and SIC project (grant \mbox{P26436-N30}) of the Austrian
  Science Fund (FWF).}
}

\begin{document}
\maketitle

\begin{textblock*}{\textwidth}(0mm,189mm)%
  \footnotesize%
  This is the unedited Author’s version of a Submitted Work that was
  subsequently accepted for publication at 2018 Design, Automation Test in
  Europe Conference Exhibition (DATE).
\end{textblock*}%
\vspace*{-1em}


\begin{abstract}

Accurate delay models are important for static and dynamic 
timing analysis of digital circuits, and mandatory for formal verification. However, F\"ugger et
al.\ [IEEE TC 2016] proved that pure and inertial delays, which are employed
for dynamic timing analysis in state-of-the-art tools like ModelSim,
NC-Sim and VCS, do not yield faithful digital circuit models.
Involution delays, which are based on delay functions
that are mathematical involutions depending on the previous-output-to-input 
time offset, were introduced by F\"ugger et al.\ [DATE'15] as a faithful alternative 
(that can easily be used with existing tools). Although involution
delays were shown to predict real signal traces reasonably accurately, 
any model with a deterministic delay function is naturally limited in its modeling
power. 

In this paper, we thus extend the involution model, by adding
non-deterministic delay variations (random or even adversarial), and prove
analytically that faithfulness is not impaired by this generalization. Albeit
the amount of non-determinism must be considerably restricted to ensure this
property, the result is surprising: the involution model differs from
non-faithful models mainly in handling fast glitch trains, where small delay
shifts have large effects. This originally suggested that adding even small
variations should break the faithfulness of the model, which turned out not to
be the case. Moreover, the results of our simulations also confirm that this
generalized involution model has larger modeling power and, hence,
applicability.

\end{abstract}

\section{Introduction}

Modern digital circuit design relies heavily on fast functional simulation tools
like Cadence NC-Sim, Mentor Graphics ModelSim or Synopsis VCS, which also allow
dynamic timing validation using suitable delay models.  In fact, for modern VLSI
technologies with their switching times in the picosecond range, static timing
analysis may not be sufficient for critical parts of a circuit, where e.g.\ the
presence of glitch trains may severely affect correctness and power
consumption. Fully-fledged analog simulations, on the other hand, are often too
costly in terms of simulation time.

Delay models like CCSM~\cite{Syn:CCSM} and ECSM~\cite{Cad:ECSM} used in
gate-level timing analysis tools make use of elaborate characterization
techniques, which incorporate technology-dependent information like driving
strengths of a gate for a wide range of voltages and load capacitances.  Based
on these data, dynamic timing analysis tools compute the delay for each gate and
wire in a specific circuit, which is then used to parametrize pure and/or
inertial delay \emph{channels} (i.e., model components representing delays).
Recall that pure delay channels model a constant transport delay, whereas
inertial delay channels \cite{Ung71} allow an input transition to proceed to its
output only if there is no subsequent (opposite) input transition within some
time window $\Delta>0$. Subsequent simulation and dynamic timing analysis runs
use these pre-computed delays as \emph{constants}, i.e., they are not
reevaluated at every point in time.

More accurate simulation and dynamic timing analysis results can be achieved by
the \emph{Degradation Delay Model (DDM)}, introduced by Bellido-D\'iaz
et~al.~\cite{BDJCAVH00,BJV06}, which allows channel delays to vary and covers
\emph{gradual} pulse cancellation effects.

F\"ugger et al.~\cite{FNS16:ToC} investigated the \emph{faithfulness} of digital
circuit models, i.e., whether a problem solvable in the model can be solved with
a real physical circuit and vice versa. Unfortunately, however, they proved that
none of the existing models is faithful: for the simple \emph{Short-Pulse
  Filtration} (SPF) problem, which resembles a one-shot variant of an inertial
delay channel, they showed that every model based on \emph{bounded
  single-history channels} (see below for the definition) either contradicts the
unsolvability of SPF in bounded time or the solvability of SPF in unbounded time
by physical circuits~\cite{Mar81}.

Single-history channels allow the input-to-output delay for a given input
transition to depend on the time of the \emph{previous} output transition.
Formally, a single-history channel is defined by a \emph{delay function}
$\delta: \IR \to \IR$, where $\delta(T)$ determines the delay of an input
transition at time~$t$, given that the previous output transition occurred at
time~$t-T$. \figref{fig:shc} depicts the involved parameters.  Note that $T$ and
$\delta(T)$ are potentially negative in the case of a short input pulse, where a
new input transition occurs earlier than the just scheduled previous output
transition.  Together with the rule that non-FIFO transitions cancel each other,
this allows to model attenuation and even suppression of glitches.
\figref{fig:pulsetrain_symb} shows an example input/output-trace generated by a
single-history channel. Note that, for \emph{bounded} single-history channels,
$\delta(T)$ cannot point arbitrarily far back into the past.

\begin{figure}
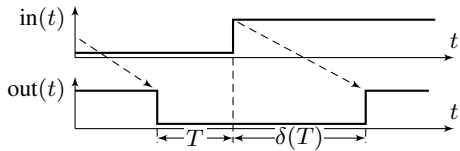

    \centering
    \tikzfigurechannelintro
    \caption{\small\em
Input/output signal of single-history channel, involving the
    previous-output-to-input delay $T$ and input-to-output delay
    $\delta(T)$.}\label{fig:shc}
\vspace{-0.3cm}
\end{figure}

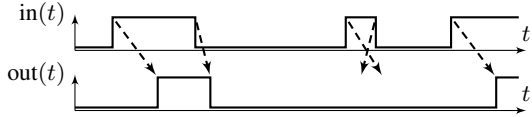
\begin{figure}
\centering
\begin{tikzpicture}[>=latex',yscale=0.8,xscale=1]
\draw[->] (0,0) -- ++(6,0) node[above] {\small $t$};
\draw[->] (0,0) -- ++(0,0.6) node[left,xshift=0pt] {\small $\text{in}(t)$};

\draw[thick] (0,0.05) -- ++(0.5,0) -- ++(0,0.5)
                      -- ++(1.1,0) -- ++(0,-0.5)
                      -- ++(2,0) -- ++(0,0.5)
                      -- ++(0.4,0) -- ++(0,-0.5)
                      -- ++(1.0,0) -- ++(0,0.5) -- ++(0.9,0);

\draw[->] (0,-1) -- ++(6,0) node[above] {\small $t$};
\draw[->] (0,-1) -- ++(0,0.6) node[left,xshift=0pt] {\small $\text{out}(t)$};

\draw[thick] (0,-0.95) -- ++(1.1,0) -- ++(0,0.5) 
                      -- ++(0.7,0) -- ++(0,-0.5)    
                      -- ++(3.8,0) -- ++(0,0.5) -- ++(0.3,0); 

\draw[densely dashed,->,shorten >=1pt,shorten <=2pt,thick] (0,0.05)++(0.5,0.5) -- (1.1,-0.45);                 
\draw[densely dashed,->,shorten >=1pt,shorten <=2pt,thick] (0,0.05)++(1.6,0.5) -- (1.8,-0.45);    

\draw[densely dashed,->,shorten >=1pt,shorten <=2pt,thick] (0,0.05)++(3.6,0.5) -- (4.1,-0.45);                 
\draw[densely dashed,->,shorten >=1pt,shorten <=2pt,thick] (0,0.05)++(4,0.5) -- (3.8,-0.45);

\draw[densely dashed,->,shorten >=1pt,shorten <=2pt,thick] (0,0.05)++(5,0.5) -- (5.6,-0.45);

\end{tikzpicture}
\caption{\small\em Single-history channels allow to model pulse attenuation:
  The delay $\delta(T)$ becomes smaller with smaller previous-output-to-input 
time $T$. Observe the cancellation of the second pulse due to non-FIFO-scheduled
  output transitions.}\label{fig:pulsetrain_symb}
\vspace{-0.4cm}
\end{figure}

In \cite{FNNS15:DATE}, F\"ugger et~al.\ introduced an \emph{unbounded}
single-history channel model based on \emph{involution channels}, which use a
delay function $\delta(T)$ whose negative is self-inverse, i.e., fulfills the
involution property $-\delta(-\delta(T))=T$. They proved that, in sharp contrast
to bounded single-history channels, SPF cannot be solved in bounded time with
involution channels, whereas it is easy to provide an unbounded SPF
implementation, which is in accordance with real physical circuits~\cite{Mar81}.
Hence, binary-valued circuit models based on involution channels are faithful
with respect to the SPF problem. We note that this actually implies faithfulness
also w.r.t.\ other, practically more relevant problems: analogous to
\cite{BJ83}, it is possible to implement a one-shot version of a latch (that
allows a single up- and a single down-transition of the enable input) using a
circuit solving SPF, and vice versa. Consequently, the involution model is also
faithful for one-shot latches.  Moreover, in \cite{NFNS15:GLSVLSI}, Najvirt et
al.\ used both measurements and Spice simulations to show that the involution
model can also be made reasonably accurate by suitable parametrization, in the
sense that it nicely (though not perfectly) predicts the actual glitch
propagation behavior of a real circuit, namely, an inverter chain.

As it is easy to replace the standard pure or inertial delays 
currently used in VITAL or Verilog models by involution delays,
the model is not only a promising starting point for sound 
formal verification, but also allows to seamlessly improve existing
dynamic timing analysis tools.


\noindent
\textbf{Main contributions:} Notwithstanding its superiority with respect to
faithfulness, like every \emph{deterministic} delay model, the involution model
has limited modeling power: many different effects in physical circuits cause
various types of noise in signal waveforms and, hence, \emph{jitter} in the
digital abstraction \cite{CR16:arxiv}.  No deterministic delay function can
properly capture the resulting variability in the signal traces.

In this paper, we relax the involution model introduced in \cite{FNNS15:DATE} by
adding limited \emph{non-determinism} $\boldsymbol{\eta}=[-\eta^-,\eta^+]$, for
some fixed $\eta^-, \eta^+ \geq0$, on top of the (deterministic) involution
delay function $\delta(T)$.  We prove that this can be done without sacrificing
faithfulness: both the original SPF impossibility result and, in particular, a
novel SPF possibility hold for this generalized model. We need to stress,
however, that adding non-determinism is merely a convenient way of securing
maximum generality of our results: no practically observable bounded jitter
phenomenon, neither bounded random noise, from white to slowly varying flicker
noise \cite{CR16:arxiv}, nor even \emph{adversarially} chosen transition time
variations can invalidate the faithfulness of the resulting \emph{\etamodel}.
Deterministic effects, like slightly different thresholds due to process
variations, are of course also covered.

Note carefully that, albeit the non-determinism ($\eta^+$ and $\eta^-$) must be
restricted to ensure faithfulness, the mere fact that we can afford some
non-determinism here \emph{at all} is very surprising: comparing the faithful
original involution model and the non-faithful DDM model reveals that they
primarily differ in handling fast glitch trains, where small delay shifts have
large effects. We thus conjectured originally that adding even small
non-determinism would break the border between both models, which we now know is
not the case.

Our generalization also results in an improved \emph{principal}\footnote{We
  stress that we do not aim at resolving the non-determinism of the \etamodel\
  to build an accurate simulator in this paper, but rather at providing a model
  that makes this possible.} modeling accuracy of the \etamodel: thanks to the
additional freedom for choosing transition times provided by
$\boldsymbol{\eta}$, it is obviously easier to match the real behavior of a
circuit with some feasible behavior of the circuit in the model.  We provide
some simulation results (in a similar setting as used in \cite{NFNS15:GLSVLSI}),
which demonstrate that it is indeed possible to match the behavior of a real
inverter chain with the \etamodel\ if the variations of operating conditions
resp.\ process variations are small. Whereas this does not hold for larger
variations, we observed that excessive deviations occur for relatively large
values of $T$ only, which are essentially irrelevant for faithfulness. We are of
course aware that more validation experiments, with more complex circuits, will
be needed to actually claim good accuracy of the \etamodel, nevertheless, our
preliminary results are encouraging.


Regarding applicability, we consider the \etamodel\ interesting for primarily
two reasons: First, it facilitates accurate modeling and analysis of circuits
under (restricted) noise, varying operating conditions and parameter variations.
Second, to the best of our knowledge, it is the first model that appears to be a
suitable basis for the sound formal verification of a circuit, which aims at
proving that the circuit meets its specification in \emph{every} feasible
trace. We thus believe that our \etamodel\ might eventually turn out to be an
interesting ingredient for a novel verification tool.


\noindent
{\bf Paper organization:}
In Section~\ref{sec:model}, we provide some indispensable basics of standard
involution channels taken from \cite{FNNS15:DATE}. Section~\ref{sec:choice}
defines our \etamodel, Section~\ref{sec:possibility} provides the proofs for
faithfulness.  Our simulation results are presented in
Section~\ref{sec:simulations}, and some conclusions and directions of our
current/future work are appended in Section~\ref{sec:conclusions}.

\section{The Involution Model Without Choice}\label{sec:model}

Before we can present the generalized \etamodel\ with non-deterministic delay
variations, we recall the basics from the circuit model introduced
in~\cite{FNNS15:DATE}.

\smallskip

\bfno{Signals.}  A {\em falling transition\/} at time~$t$ is the pair $(t,0)$, a
{\em rising transition\/} at time~$t$ is the pair $(t,1)$.  A {\em signal\/} is
a list of alternating \ifthenelse{\boolean{SHORTversion}}{transitions where only
  finitely many alternating transitions may occur in a finite time interval.}{
  transitions such that
\begin{enumerate}[S1)]
\item the initial transition is at time~$-\infty$; all other transitions are at times~$t\geq0$,
\item the sequence of transition times is strictly increasing,
\item if there are infinitely many transitions in the list, then the set of
transition times is unbounded.
\end{enumerate}

To every signal~$s$ (uniquely) corresponds a function $\IR\to\{0,1\}$, its \emph{signal trace}, whose value at time~$t$ is that of the most recent transition.
}

\smallskip

\bfno{Circuits.}  Circuits are obtained by interconnecting the external
interface, i.e., a set of input and output ports, and a set of combinational
gates via channels.  The valid connections are constrained by demanding that
gates and channels must alternate on every path in the circuit and that any gate
input and output port is attached to only one channel output.  Formally we
describe a circuit by a directed graph with potentially multiple edges between
nodes.  Its nodes are in/out ports and gates, and edges are channels.  A channel
has a channel function, which maps input signals to output signals, whereas a
gate is characterized by a (zero-time) Boolean function and an initial Boolean
value that defines its output until time~$0$.  Channels connecting input and
output ports are assumed to have zero delay, in order to facilitate the
composition of circuits.

\smallskip

\bfno{Executions.}  An execution of circuit~$C$ is an assignment of signals to
the vertices and edges of $C$ that respects channel functions, Boolean gate
functions, and initial values of gates.  Signals on input ports are
unrestricted.  For an edge $c$ representing a channel with channel function
$f_c$ from vertex $v$ in $C$, we require that the signal $s_c$ assigned to $c$
fulfills $s_c = f_c(s_v)$.

\smallskip

\bfno{Involution Channels.}  An involution channel propagates each transition at
time~$t$ of the input signal to a transition at the output happening after some
input-to-output delay $\delta(T)$, which depends on the previous-output-to-input
delay $T$ (cf.\ \figref{fig:shc}).

An involution channel function is characterized by
two strictly increasing concave delay
functions
$\delta_\uparrow:(-\delta^\downarrow_\infty,\infty)\to(-\infty,\delta^\uparrow_\infty)$
and
$\delta_\downarrow:(-\delta^\uparrow_\infty,\infty)\to(-\infty,\delta^\downarrow_\infty)$
such that both 
$\delta^\uparrow_\infty=\lim_{T\to\infty}\delta_\uparrow(T)$ 
and
$\delta^\downarrow_\infty=\lim_{T\to\infty}\delta_\downarrow(T)$ 
are finite and
\begin{equation}\label{eq:involution}
-\delta_\uparrow\big( -\delta_\downarrow(T) \big) = T
\text{ and }
-\delta_\downarrow\big( -\delta_\uparrow(T) \big) = T
\end{equation}
for all~$T$.  All such functions are necessarily continuous.  For simplicity, we
will also assume them to be differentiable; $\delta$ being concave thus implies
that its derivative $\delta'$ is monotonically decreasing.  In this paper, we
assume all involution channels to be {\em strictly causal}, i.e.,
$\delta_\uparrow(0)>0$ and $\delta_\downarrow(0)>0$.

A particular and important special case are the so-called {\em exp-channels}:
They occur when gates drive RC-loads and generate digital transitions when reaching a 
certain threshold voltage $V_{th}$ (typically $V_{th} = 1/2$ of the maximum
voltage $V_{DD}$).
We obtain
\begin{eqnarray*}
\delta_\uparrow(T)=\tau\ln(1-e^{-(T+T_p-\tau\ln(\overline{V_{th}}))/\tau})+T_p-\tau\ln(1-\overline{V_{th}})\nonumber\\
\delta_\downarrow(T)=\tau\ln(1-e^{-(T+T_p-\tau\ln(1-\overline{V_{th}}))/\tau})+T_p-\tau\ln(\overline{V_{th}})\,,\,\label{eq:exp}
\end{eqnarray*}
where~$\tau$ is the RC constant,~$T_p$
  the pure delay component and~$\overline{V_{th}} = V_{th}/V_{DD}$.

\ifthenelse{\boolean{SHORTversion}}{}{
For ease of reference, we restate the following technical lemma
from \cite{FNNS14:arxiv,FNNS15:DATE}:

\begin{lemma}
\label{lem:delta:min}
A strictly causal involution channel has a unique~$\deltamin$ defined
by $\delta_\uparrow(-\deltamin) = \deltamin = 
\delta_\downarrow(-\deltamin)$, which is positive.
For exp-channels, $\deltamin=T_p$.

For the derivative, we have $\delta_\uparrow'(-\delta_\downarrow(T)) = 1/\delta_\downarrow'(T)$
and hence $\delta_\uparrow'(-\deltamin)  =
1/\delta_\downarrow'(-\deltamin)$.
\end{lemma}
}

The channel function $f_c$ mapping input signal $s$ to output signal $f_c(s)$
(cp.\ \figref{fig:pulsetrain_symb}) is defined via the following algorithm. It
can easily be implemented in e.g.\ VHDL to be used by existing simulators like
ModelSim, as these simulators automatically drop transitions on signals
violating FIFO order.

{\em Output transition generation algorithm:\/}
Let~$t_1,t_2,\dots$ be the transitions times of~$s$,
set $t_0 = -\infty$ and $\delta_0=0$.
\begin{itemize}
\item {\em Initialization:\/}
Copy the initial transition at time $-\infty$ from the input signal
  to the output signal.
\item {\em Iteration:\/} Iteratively determine the tentative list of pending output
  transitions:
Determine the input-to-output delay $\delta_n$ for the input transition at time~$t_n$ by
setting
$\delta_n = \delta_\uparrow(t_n - t_{n-1} - \delta_{n-1})$ if $t_n$ is a rising
transition and
$\delta_n = \delta_\downarrow(t_n-t_{n-1}-\delta_{n-1})$
if it is falling.
The $n$\textsuperscript{th} and $m$\textsuperscript{th} pending output transitions {\em cancel\/} if $n < m$ but
$t_n+\delta_n \geq t_m+\delta_m$. In this case, we mark both as canceled.
\item {\em Return:\/}
The channel output signal~$f_c(s)$ has the same initial value as the input signal, and contains every pending 
transition at time $t_n + \delta_n$ that has not been marked as canceled.
\end{itemize}

\section{Introducing Adversarial Choice}\label{sec:choice}

We now generalize the circuit model from the previous section to allow a
non-deterministic perturbation of the output transition times after the
application of the delay functions~$\dup$ and~$\ddo$.  Note that the resulting
output shifts need \emph{not} be the same for all applications of the delay
functions; they can vary arbitrarily from one transition to the next.  However,
each perturbation needs to be within some pre-determined interval
${\boldsymbol{\eta}} = [-\eta^-,\eta^+]$.  These non-deterministic choices can
be used to model various effects in digital circuits that cannot be captured by
single-history delay functions, ranging from arbitrary types of noise
\cite{CR16:arxiv} to unknown variations of process parameters and operating
conditions.
\begin{figure}
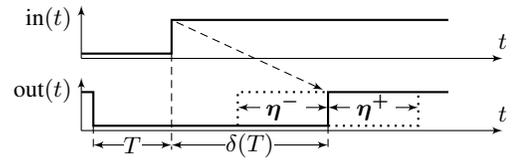

\center
\tikzfigureintroeta2
\caption{\small\em The \etachannel: Non-deterministic choice of the tentative output transition after applying $\delta(T)$.} 
\label{fig:delta:eta}
\vspace{-0.5cm}
\end{figure}
Fig.~\ref{fig:delta:eta} shows the possible variation of the output transition time
caused by the non-deterministic choice.

Formally, we change the notion of the \emph{channel function} to accept an
additional parameter: A channel has a channel function, which maps each pair
$(s,H)$ to an output signal, where~$s$ is the channel's input signal and~$H$ is
a parameter taken from some suitable set of admissible parameters (see below).
We also adapt the definition of an \emph{execution} to allow an adversarial
choice of~$H$: For an edge $c$ from $v$ in $C$, we require that there exists
some admissible parameter $H$ such that the signal $s_c$ fulfills
$s_c = f_c(s_v,H)$.

For \etachannel s, we let the admissible parameters~$H$ be any sequence of
choices~$\eta_n \in {\boldsymbol{\eta}}$. The output transition generation
algorithm's \emph{Iteration} step for the $n$\textsuperscript{th} transition of
the input signal is adapted as follows:
$\delta_n = \delta_\uparrow(\max\{t_n - t_{n-1} -
\delta_{n-1},-\delta_\infty^\downarrow\})+\eta_n$ if $t_n$ is a rising transition
and
$\delta_n =
\delta_\downarrow(\max\{t_n-t_{n-1}-\delta_{n-1},-\delta_\infty^\uparrow\})+\eta_n$
if it is falling. Note that the $\max$-terms guard agains adversarial choices
that would exceed the domain of $\dup(.)$ and $\ddo(.)$. This could occur only
in the extreme situation of a short glitch after a long stable input, which must
be canceled anyway. So enforcing
$\delta_n = \delta_\uparrow(-\delta_\infty^\downarrow)+\eta_n=-\infty$ resp.\
$\delta_n = \delta_\downarrow(-\delta_\infty^\uparrow)+\eta_n = -\infty$ in
this case is safe. As this cannot occur in the cases analyzed in this paper, we
will subsequently omit the $\max$-terms in the definition of $\delta_n$ for
simplicity.

\figref{fig:pulsetrain_eta_symb} depicts two example signal traces,
$\text{out}_1$ and $\text{out}_2$, obtained by an \etachannel\ with the same
underlying $\delta$ as the one in \figref{fig:pulsetrain_symb}.  Observe that
the adversary has the freedom to ``de-cancel'' pulses that would have canceled
according to the delay function (second pulse in $\text{out}_2$), extend pulses
(first pulse in $\text{out}_1$), and shift pulses (first pulse in
$\text{out}_2$).

\begin{figure}
\centering
\begin{tikzpicture}[>=latex',yscale=0.8]
\draw[->] (0,0) -- ++(6,0) node[above] {\small $t$};
\draw[->] (0,0) -- ++(0,0.6) node[left] {\small $\text{in}(t)$};

\draw[thick] (0,0.05) -- ++(0.5,0) -- ++(0,0.5)
                      -- ++(1.1,0) -- ++(0,-0.5)
                      -- ++(2,0) -- ++(0,0.5)
                      -- ++(0.4,0) -- ++(0,-0.5)
                      -- ++(1.0,0) -- ++(0,0.5) -- ++(0.9,0);
                      
\draw[->] (0,-1) -- ++(6,0) node[above] {\small $t$};
\draw[->] (0,-1) -- ++(0,0.6) node[left] {\small $\text{out}_1(t)$};

\draw[dotted,thick] (0,-0.95) -- ++(1.1,0) -- ++(0,0.5) 
                      -- ++(0.7,0) -- ++(0,-0.5)    
                      -- ++(3.8,0) -- ++(0,0.5) -- ++(0.3,0); 

\draw[thick] (0,-0.95) -- ++(0.8,0) -- ++(0,0.5)
                      -- ++(1.5,0) -- ++(0,-0.5)
                      -- ++(2.5,0) -- ++(0,0.5) -- ++(1.1,0);                      

\draw[densely dashed,->,shorten >=1pt,shorten <=2pt,thick] (0,0.05)++(0.5,0.5) -- (1.1,-0.45);                 
\draw[densely dashed,->,shorten >=1pt,shorten <=2pt,thick] (0,0.05)++(1.6,0.5) -- (1.8,-0.45);    

\draw[densely dashed,->,shorten >=1pt,shorten <=2pt,thick] (0,0.05)++(3.6,0.5) -- (4.1,-0.45);                 
\draw[densely dashed,->,shorten >=1pt,shorten <=2pt,thick] (0,0.05)++(4,0.5) -- (3.8,-0.45);

\draw[densely dashed,->,shorten >=1pt,shorten <=2pt,thick] (0,0.05)++(5,0.5) -- (5.6,-0.45);

\draw[->] (0,-1.1)++(1.1,0) -- ++(-0.3,0) node[midway,below] {\small -$\eta_1$}; 
\draw[->] (0,-1.1)++(1.8,0) -- ++(0.5,0) node[midway,below] {\small $\eta_2$}; 
\draw ($ (0,0.05)+(3.6,-1.15) $) node[below] {\small $\eta_3$=$\eta_4$=$0$}; 
\draw[->] (0,-1.1)++(5.6,0) -- ++(-0.8,0) node[midway,below] {\small -$\eta_5$}; 

\end{tikzpicture}\\
\begin{tikzpicture}[>=latex',yscale=0.8]
\draw[->] (0,0) -- ++(6,0) node[above] {\small $t$};
\draw[->] (0,0) -- ++(0,0.6) node[left] {\small $\text{in}(t)$};

\draw[thick] (0,0.05) -- ++(0.5,0) -- ++(0,0.5)
                      -- ++(1.1,0) -- ++(0,-0.5)
                      -- ++(2,0) -- ++(0,0.5)
                      -- ++(0.4,0) -- ++(0,-0.5)
                      -- ++(1.0,0) -- ++(0,0.5) -- ++(0.9,0);

\draw[->] (0,-1) -- ++(6,0) node[above] {\small $t$};
\draw[->] (0,-1) -- ++(0,0.6) node[left] {\small $\text{out}_2(t)$};

\draw[dotted,thick] (0,-0.95) -- ++(1.1,0) -- ++(0,0.5) 
                      -- ++(0.7,0) -- ++(0,-0.5)    
                      -- ++(3.8,0) -- ++(0,0.5) -- ++(0.3,0); 

\draw[thick] (0,-0.95) -- ++(1.4,0) -- ++(0,0.5)
                       -- ++(0.8,0) -- ++(0,-0.5)
                       -- ++(1.3,0) -- ++(0,0.5)
                       -- ++(0.9,0) -- ++(0,-0.5)
                       -- ++(0.8,0) -- ++(0,0.5) -- ++(0.7,0);

\draw[densely dashed,->,shorten >=1pt,shorten <=2pt,thick] (0,0.05)++(0.5,0.5) -- (1.1,-0.45);                 
\draw[densely dashed,->,shorten >=1pt,shorten <=2pt,thick] (0,0.05)++(1.6,0.5) -- (1.8,-0.45);    

\draw[densely dashed,->,shorten >=1pt,shorten <=2pt,thick] (0,0.05)++(3.6,0.5) -- (4.1,-0.45);                 
\draw[densely dashed,->,shorten >=1pt,shorten <=2pt,thick] (0,0.05)++(4,0.5) -- (3.8,-0.45);

\draw[densely dashed,->,shorten >=1pt,shorten <=2pt,thick] (0,0.05)++(5,0.5) -- (5.6,-0.45);

\draw[->] (0,-1.1)++(1.1,0) -- ++(0.3,0) node[midway,below] {\small $\eta_1$}; 
\draw[->] (0,-1.1)++(1.8,0) -- ++(0.4,0) node[midway,below] {\small $\eta_2$}; 
\draw[->] (0,-1.1)++(4.1,0) -- ++(-0.6,0) node[below] {\small -$\eta_3$}; 
\draw[->] (0,-1.5)++(3.8,0) -- ++(0.6,0) node[above] {\small $\eta_4$}; 
\draw[->] (0,-1.1)++(5.6,0) -- ++(-0.4,0) node[midway,below] {\small -$\eta_5$};

\draw[dashed] (3.8,-0.45) -- (3.8,-1.5);
\draw[dashed] (4.1,-0.45) -- (4.1,-1.1);

\end{tikzpicture}
\ifthenelse{\boolean{SHORTversion}}{\caption{\small\em The \etachannel\ covers pulse attenuation under (bounded) adversarial noise,
  varying operating conditions, parameter variations and other modeling inaccuracies.
}\label{fig:pulsetrain_eta_symb}}
{\caption{\small\em The \etachannel\ covers pulse attenuation under (bounded) adversarial noise,
  varying operating conditions, parameter variations and other modeling inaccuracies.
Observe the different output behaviors $\text{out}_1$ and $\text{out}_2$ for the same input trace, caused by different adversarial choices ($\eta_1, \eta_2, \dots$). 
The output transitions that would have been caused just by $\delta(T)$, without $\boeta$-shifts, 
are dotted. Note that different adversarial choices usually change the history and, hence, $T$ 
and thus $\delta(T)$.
}\label{fig:pulsetrain_eta_symb}}
\end{figure}

\section{Faithfulness of Involution Channels with Adversarial Choice}\label{sec:possibility}

In this section, we will prove that \etachannel s are faithful with respect to
{\em Short-Pulse Filtration (SPF)\/}.

A {\em pulse\/} of length~$\Delta$ at time~$T$ has initial value~$0$, one rising
transition at time~$T$, and one falling transition at time $T+\Delta$.
\ifthenelse{\boolean{SHORTversion}}{}{A signal {\em contains a pulse\/} of
  length~$\Delta$ at time~$T$ if it contains a rising transition at time~$T$, a
  falling transition at time $T+\Delta$ and no transition in between.}

\begin{definition}[Short-Pulse Filtration]\label{def:spf}
A circuit with a single input and a single output port solves Short-Pulse
Filtration (SPF), if it fulfills the
following conditions for all admissible channel function parameters~$H$:
\begin{enumerate}[F1)]
\item The circuit has exactly one input and one output port. {\em
        (Well-formedness)}
\item A zero input signal produces a zero output
	signal. {\em (No generation)}
\item There exists an input pulse such that
	the output signal is not the zero signal. {\em (Nontriviality)}
\item There exists an~$\varepsilon>0$ such that for every input pulse the output
	signal never contains a pulse of length less than~$\varepsilon$. {\em
	(No short pulses)}
\end{enumerate}
\end{definition}
Note that we allow the SPF circuit to behave arbitrarily if the input
signal is not a (single) pulse.

To show faithfulness of the \etamodel, we start with the trivial direction: we
prove that no circuit with \etachannel s can solve the bounded-time variant of
SPF (where the output must stabilize to constant 0 or 1 within bounded
time). Note that this matches the well-known impossibility \cite{Mar77} of
building such a circuit in reality.  Indeed, the result immediately follows from
the fact that the adversary is free to always choose $\eta_n = 0$, i.e., make
the \etachannel s behave like involution channels.  In
\cite{FNNS15:DATE,FNNS14:arxiv}, it has been shown that no circuit with
involution channels can solve bounded-time SPF, which completes the proof.

What hence remains to be shown is the existence of a circuit that solves SPF
(with unbounded stabilization time) with \etachannel s.  We can prove that the
circuit shown in \figref{fig:circuit}, which consists of a fed back OR-gate
forming the storage loop and a subsequent buffer with a suitably chosen (high)
threshold voltage (modeled as an exp-channel), does the job.  As a consequence,
a circuit model based on \etachannel s enjoys the same faithfulness as the
involution channels of \cite{FNNS15:DATE}, even though its set of allowed
behaviors is considerably larger.

\begin{figure}
\centering
\begin{tikzpicture}[>=latex',circuit logic US, scale=1.0]

\draw (0,0) node[or gate,small circuit symbols] (nor) {\,\,\,OR};

\draw (nor.output) -- ++(0.1,0)
  -- ++(0,-0.7)
  -- ++(-1.1,0) node[fill=white,right,draw=black,rectangle,rounded corners,minimum width=10mm,minimum height=4.5mm,xshift=-0mm]
  (chan) {$c$}
  -- ++(-0.2,0)
  |- (nor.input 2);

\draw[<-] (nor.input 1)
        -- ++(left:0.8) node[left] (i) {$i$};

\draw[->] (nor.output) -- ++(0.1,0) node[circle,inner sep=0.7pt,fill=black,draw] (huhu) {}
        -- ++(right:0.2) node[fill=white,right,draw=black,rectangle,rounded corners,minimum width=10mm,minimum height=4.5mm,xshift=-0mm] (ht) {HT}
        -- ++(right:1.4) node[buffer gate,small circuit symbols,fill=white] {}
	-- ++(right:0.9) node[right] {$o$}
	;

\end{tikzpicture}
\caption{\small\em A circuit solving unbounded SPF, consisting of an OR-gate, with initial value $0$, fed back by
channel~$c$, and a high-threshold buffer HT.}
\label{fig:circuit}
\vspace{-0.5cm}
\end{figure}
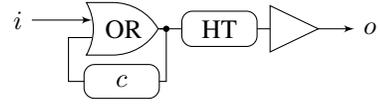

Informally, we consider a pulse of length~$\Delta_0$ at time~$0$ at the input
and reason about the behavior of the feed-back loop, i.e., the output of the OR
gate.  There are 3 cases: If $\Delta_0$ is small, then the pulse is filtered by
the channel in the feed-back loop.  If it is big, the pulse is captured by the
storage loop, leading to a stable output~1.  For a certain range of $\Delta_0$,
the storage loop may be oscillating, possibly forever. In any case, however, it
turns out that a properly chosen exp-channel can translate this behavior to a
legitimate SPF output.

\ifthenelse{\boolean{SHORTversion}}{}
{
\begin{lemma}
\label{lem:big:pulse}
If the input pulse's length $\Delta_0$ satisfies
$\Delta_0\geq\delta_\infty^\uparrow + \eta^+$, then the output of the OR in
\figref{fig:circuit} has a unique rising transition at time~0, and no falling
transition.
\end{lemma}
\begin{proof}
  Clearly, the output of the OR, hence the \etachannel's input, will have a
  rising transition at time 0.  The corresponding rising transition occurs at
  the channel output at the latest at
  $\eta^+ + \delta_\infty^\uparrow \leq \Delta_0$.  This guarantees the storage
  loop to lock, causing the output of the OR output to stick to~1.
\end{proof}

\begin{lemma}
\label{lem:small:pulse}
If the input pulse's length $\Delta_0$ satisfies
$\Delta_0 \leq \delta_\infty^\uparrow - \deltamin - \eta^+ - \eta^-$, then the OR output 
in \figref{fig:circuit} contains only the input pulse.
\end{lemma}
\begin{proof}
  The input signal contains only two transitions: one at time $t_1=0$ and one at
  time $t_2=\Delta_0$. The earliest time when the output transition
  corresponding to the rising input transition can occur is
  $t_1' = \delta_\infty^\uparrow - \eta^-$.  For the falling input transition,
  we thus get $T=\Delta_0-\delta_\infty^\uparrow + \eta^-$, and observe that the
  corresponding falling output transition cannot occur later than
  $t_2'=\Delta_0+\eta^+ + \ddo(T)$. The two output transitions cancel iff
  $t_2' \leq t_1'$, which is equivalent to
  $X=\Delta_0+\eta^+ + \ddo(T)-\delta_\infty^\uparrow + \eta^- \leq
  0$. Replacing $\Delta_0$ with the upper bound from the lemma reveals
  $T \leq -\deltamin-\eta^+$ and
  $X \leq -\deltamin + \ddo(-\deltamin-\eta^+) \leq -\deltamin +
  \ddo(-\deltamin) = 0$ by monotonicity of $\ddo$ and Lemma~\ref{lem:delta:min},
  which concludes the proof.
\end{proof}

For an input pulse length that satisfies
$\delta_\infty^\uparrow-\deltamin- \eta^+ - \eta^- < \Delta_0<
\delta_\infty^\uparrow+ \eta^+$, the OR output signal may contain a series of
pulses of lengths $\Delta_0,\Delta_1,\Delta_2, \dots$. In sharp contrast to
standard involution channels \cite{FNNS15:DATE}, it is \emph{not} the case that
there is a unique value~$\Delta_0=\tilde{\Delta}_0$ that leads to an infinite
series of (identical) pulses $\Delta_1=\Delta_2=\dots$\ Rather, due to the
adversarial choices, there is a range of values for $\Delta_0$ that may lead to
a whole range of infinite pulse trains, with varying pulse lengths, which are
surprisingly difficult to bound.

An informal, high-level explanation of the approach that was eventually found to
be successful is the following: we identified a self-repeating infinite
``worst-case pulse train'', which ensures that any adversarial choice that
deviates from it at some point causes the subsequent pulses to die out, i.e., to
resolve to a stable 1. In more detail, let $\Delta_0$ be such that an infinite
self-repeating pulse train $\Delta=\Delta_1=\Delta_2=\dots$ exists, subject to
the constraint that the adversary deterministically takes all rising transitions
maximally ($\eta^+$) late and all falling transitions maximally ($\eta^-$)
early.  Note that this adversarial choice actually minimizes $\Delta_n$ for any
given $\Delta_{n-1}$. Therefore, given a pulse $\Delta_{n-1}=\Delta$, any other
adversarial choice (as well as any larger $\Delta_{n-1}>\Delta$) leads to a
subsequent pulse with $\Delta_n > \Delta$. As a consequence, $\Delta$ is an
\emph{upper} bound for the length of \emph{every} pulse $\Delta_n$, $n\geq 1$,
occurring in an arbitrary \emph{infinite} pulse train: if some
$\Delta_{n-1}>\Delta$ ever happens, then $\Delta_{n+\ell} > \Delta$ for every
$\ell\geq 0$ as well; in fact, Lemma~\ref{lem:exponential:decay} will reveal
that the pulse train will only be finite in these cases.

Similarly, since the adversarial choice that minimizes the up-time $\Delta_n$
simultaneously maximizes the down-time $\Deltado_n$ of a pulse, we also get a a
lower bound $\Deltado_n \geq P-\Delta$ for all pulses in an arbitrary infinite
pulse train, where $P$ is the period of our infinite self-repeating pulse train.

For these arguments to work, we need to restrict the adversarial choice
for the feed-back channel in \figref{fig:circuit}:
\begin{equation}
\eta^+ + \eta^- < \ddo(-\eta^+)-\deltamin\tag{C}\label{C}
\end{equation}
Formally, we have the following 
Lemma~\ref{lem:mid:pulse}:

\begin{lemma}
\label{lem:mid:pulse}
Consider the circuit in \figref{fig:circuit} subject to constraint \eqref{C}.
Assume that the input pulse length $\Delta_0$ is such that it results in an
infinite pulse train $\Delta_0, \Delta_1, \dots$ occurring at the output of the
OR.  Then, for every $n\geq 1$, the up-time $\Delta_n$ satisfies
$\Delta_n \leq \Delta$, the down-time $\Delta_n'$ (preceding the pulse with
up-time $\Delta_n$) satisfies $\Delta_n' \geq P-\Delta$, and
$P_n=\Delta_n+\Delta_{n+1}' \geq P$. Herein, $\Delta=\ddo(\eta^+-\tau)$ with
$\Delta < \deltamin$ is the up-time of an infinite self-repeating pulse train
with period $P=\tau$ and duty cycle $\gamma=\Delta/P$, with $\tau>0$ denoting
the smallest positive fixed point of the equation
$\ddo(\eta^+ - \tau) + \dup(-\eta^- - \tau) = \tau$, which is guaranteed to
exist and satisfies
$\eta^++\deltamin < \tau < \min(-\eta^- + \delta^\downarrow_\infty, \eta^+ +
\delta^\uparrow_\infty)$.
\end{lemma}
\begin{proof}
In the circuit of Figure~\ref{fig:circuit}, the  $n$\textsuperscript{th} input pulse of the \etachannel\ $c$ is
just its $(n-1)$\textsuperscript{th} output pulse. Therefore, for all $n>1$, the output pulse length
$\Delta_n$ under the worst-case
adversarial choice of $\eta^+$-late rising and $\eta^-$-early falling
transitions evaluates to
\begin{eqnarray}
\Delta_n = f(\Delta_{n-1}) &=& \ddo\big( \Delta_{n-1} -\eta^+ - \dup(-\Delta_{n-1}) \big)\label{def:delfunc}\label{def:ffunc}\\
&& {}+ \Delta_{n-1} -\eta^- - \eta^+ -
\dup(-\Delta_{n-1})\enspace.\nonumber
\end{eqnarray}
The sought fixed point $\Delta$ of (\ref{def:ffunc})
resulting in a infinite pulse train is obtained by solving  
$\Delta=f(\Delta)$, which yields
\begin{equation}
\ddo\big( \Delta -\eta^+ - \dup(-\Delta) \big) = 
\eta^- + \eta^+ +
\dup(-\Delta)\label{eq:ffix}
\enspace.
\end{equation}
Applying the involution property to (\ref{eq:ffix}) results in
$\Delta -\eta^+ - \dup(-\Delta) = -\dup(-\eta^- - \eta^+ -
\dup(-\Delta))$ and further in
\begin{equation}
\Delta +\dup\big(-\eta^- - \eta^+ -
\dup(-\Delta)\big) = \eta^+ + \dup(-\Delta)\label{eq:ffix2}
\enspace.
\end{equation}

Defining $\tau=\eta^+ + \dup(-\Delta)$, rewriting it
to $-\dup(-\Delta) = \eta^+-\tau$ and applying the involution property, we observe
\begin{equation}
\Delta=\ddo(\eta^+-\tau) \label{def:deltaprime}
\enspace.
\end{equation}
Using 
\eqref{def:deltaprime} and \eqref{eq:involution} in \eqref{eq:ffix2}
yields the fixed point equation stated in our lemma:
\begin{equation}
\ddo(\eta^+ - \tau) + \dup(-\eta^- - \tau) = \tau\label{eq:fixtau}
\enspace.
\end{equation}

Now assume that the smallest fixed point $\tau>0$ of (\ref{eq:fixtau}), and
hence $\Delta$ of (\ref{def:ffunc}), exists. Then, in any infinite pulse train,
any pulse $\Delta_{n-1}>\Delta$, $n>1$, and/or any non-worst-case adversarial
choice (also in the case $\Delta_{n-1}=\Delta$) leads to a subsequent pulse with
$\Delta_n > \Delta$. As a consequence, $\Delta$ is indeed an upper bound for the
length of \emph{every} such pulse.

\medskip

We will proceed in our proof with establishing constraints on $\eta^-$, $\eta^+$
that guarantee the existence of a solution $\tau>0$ of (\ref{eq:fixtau}).  For
this purpose, we introduce the function
\begin{equation}
h(\tau)=\ddo(\eta^+ - \tau) + \dup(-\eta^- - \tau) - \tau \label{eq:hdef} \enspace.
\end{equation}
and show that there are values $\tau_{0} < \tau_1$ where $h(\tau_0) > 0$ but
$h(\tau_1) < 0$. Since $h(.)$ is continuous, this ensures the existence of
$\tau_{0} < \tau < \tau_1$ with $h(\tau)=0$.

If we plug in $\tau_0=\eta^++\deltamin$ in (\ref{eq:hdef}), we find by recalling
Lemma~\ref{lem:delta:min} that
$h(\eta^++\deltamin)= \dup(-\eta^+ - \eta^- - \deltamin) - \eta^+$.  In order to
guarantee that $h(\eta^++\deltamin) > 0$ we need
$\dup(-\eta^+ - \eta^- - \deltamin) > \eta^+$. Rewriting this using the
involution property requires
$-\dup(-\eta^+ - \eta^- - \deltamin) < -\dup(-\ddo(-\eta^+))$ and hence
$\eta^++\eta^- < \ddo(-\eta^+) - \deltamin$ as stated in constraint \eqref{C}.
Note that this implies $\eta^+ < \deltamin$, since $\eta^++\eta^-\geq 0$.

For $h(\tau)<0$, we simply obtain $-\infty$ from $\ddo(\eta^+ - \tau)$ or
$\dup(-\eta^- - \tau)$ by plugging in
$\tau_1 = \min(-\eta^- + \delta^\downarrow_\infty, \eta^+ +
\delta^\uparrow_\infty)$ in \eqref{eq:hdef}, noting that the involution property
guarantees
$-\infty=\dup(-\delta^\downarrow_\infty)=\ddo(-\delta^\uparrow_\infty)$.  Since
all other terms of $h(.)$ are finite, the result is definitely $<0$.

We still need to assure that the boundary interval for $\tau$ is not empty,
i.e., that
$\tau_0=\eta^+ + \deltamin < \tau_1=\min(-\eta^- + \delta^\downarrow_\infty,
\eta^+ + \delta^\uparrow_\infty)$. This is trivially the case if
$\tau_1= \eta^+ + \delta^\uparrow_\infty$.  If
$\tau_1=\delta^\downarrow_\infty - \eta^-$, we need
$\eta^+ + \eta^- < \delta^\downarrow_\infty - \deltamin$, which is implied by
constraint \eqref{C}.  Thus, putting everything together, we can indeed
guarantee a solution $\tau$ of $h(\tau)=0$, which satisfies
\begin{equation}
0 < \eta^+ + \deltamin < \tau < \min(-\eta^- + \delta^\downarrow_\infty, \eta^+ + \delta^\uparrow_\infty)
\enspace
\end{equation}
as stated in our lemma.

We can now determine the upper bound for $\Delta$: Recalling the definition
$\tau=\eta^+ + \dup(-\Delta)$, the lower bound on $\tau$ implies
$\deltamin < \tau - \eta^+ = \dup(-\Delta)$.  Using the involution property, we
can translate this to $
-\ddo(-\deltamin) <
-\Delta$.

Applying Lemma~\ref{lem:delta:min}, we end up with
\begin{equation}
\Delta < \deltamin
\label{eq:Delta}
\end{equation}
as asserted in this lemma.

Regarding the periods of our pulses, we recall that our adversary takes all
rising transitions maximally late and all falling transitions maximally early to
minimize the high-times of the generated pulse train. The period
$P_{n} = \Delta_n + \Delta_{n+1}'$ of the high-pulse $\Delta_{n}$, measured from
the rising transition of $\Delta_{n}$ to the rising transition of
$\Delta_{n+1}$, is $P_{n} = \dup(-\Delta_{n})+\eta_n^+$, which is not difficult
to see from the considerations leading to (\ref{def:delfunc}).  Hence, $P_n$
only depends on the up-time $\Delta_{n}$ and the adversarial choice
$\eta_n^+ \leq \eta^+$.  It follows that the adversarial choices used for
generating our minimal up-time pulse train simultaneously \emph{maximize} both
the period ($P = \dup(-\Delta)+\eta^+$) and the down-time ($P-\Delta$).  As the
adversary cannot further shrink the up-times of the pulses, it cannot further
extend the down-times, without running into cancellations.

Formally, by the same argument as used for $\Delta$, we find that no infinite
pulse train can contain a pulse with a downtime strictly smaller than
$P-\Delta$, where $P=P'$ is the period of our infinite $\Delta$ pulse train:
analogously to $P_{n}$ above, we find that the down-period
$P_{n}'=\Delta_{n}'+\Delta_{n}$, measured between the falling transitions of
$\Delta_{n}'$ and $\Delta_{n+1}'$, evaluates to
$P_{n}'= \ddo(-\Delta_{n}')-\eta_n^-$, which decreases with both $\Delta_{n}'$
and $\eta_n^-\leq \eta^-$.  If $\Delta_n'<P-\Delta$ ever occurred, this would
lead to $P_n' > P'=\ddo(-P+\Delta)-\eta^- $. Since obviously $P'=P$, this
implies $\Delta_n = P_n' - \Delta_n' > \Delta$, which contradicts the previously
established upper bound $\Delta_n\leq \Delta$, however.

It hence only remains to evaluate $P = \dup(-\Delta)+\eta^+
= \tau$, which completes the proof.
\end{proof}

\begin{lemma}\label{lem:dutycycle}
Consider the circuit in \figref{fig:circuit} subject to constraint \eqref{C}.
The duty cycle $\gamma_n$ of any pulse $\Delta_n$, $n\geq1$, in an infinite
pulse train at the output of the OR-gate satisfies $\gamma_n \leq \gamma < 1$.
\end{lemma}
\begin{proof}
According to Lemma~\ref{lem:mid:pulse}, we have
$\gamma_{n} = \frac{\Delta_n}{P_{n}} 
\leq \frac{\Delta}{P} =
\gamma = 
\frac{\Delta}{\dup(-\Delta)+\eta^+}  <
\frac{\deltamin}{\deltamin+\eta^+}\leq 1$ 
for every $n\geq 1$ as asserted.
\end{proof}

We remark that $\eta^+>0$ allows strengthening constraint \eqref{C}, which
allows sharpening some inequalities in Lemma~\ref{lem:mid:pulse}, namely,
$\eta^+ + \eta^- \leq \ddo(-\eta^+)-\deltamin$, $\Delta \leq \deltamin$, and
$\eta^++\deltamin \leq \tau$, without violating $\gamma < 1$ established in
Lemma~\ref{lem:dutycycle}.

The following lemma implies that if $\Delta_1>\Delta$ for $\Delta$ according to
Lemma~\ref{lem:mid:pulse}, then the sequence of generated output pulses
$\Delta_n$, $n\geq 1$, will be strongly monotonically increasing. Consequently,
we will only get a bounded number of pulses at the output of the OR gate, with a
stabilization time in the order of $\log_{a} (1/(\Delta_1 - \Delta ))$ with
$a=1+\dup'(0) > 1$.

\begin{lemma}\label{lem:exponential:decay}
For $f(.)$ given in (\ref{def:ffunc}) with fixed point $\Delta$, we have
$f(\Delta_1) - \Delta \geq (1+\dup'(0)) \cdot (
\Delta_1 - \Delta)$ if $\Delta_1>\Delta$.
\end{lemma}
\begin{proof}
Differentiation of (\ref{def:ffunc}) provides
\begin{eqnarray}
f'(\Delta_1)  
&=&\big( 1 + \dup'(-\Delta_1) \big)
\Big(1+\ddo'\big(\Delta_1 -\eta^+ -
\dup(-\Delta_1)\big)\Big)\nonumber\\
&\geq& 1 + \dup'(0)
\end{eqnarray}
because $\dup'(-\Delta_1)\geq \dup'(0)$ as $\Delta_1 > \Delta > 0$ and
$\delta'(T)>0$ is decreasing for all~$T$ as $\delta(.)$ is concave and
increasing by Lemma~\ref{lem:delta:min}.  The mean value theorem of calculus now
implies the lemma.
\end{proof}

The following lemma allows to extend the validity of the statement of 
Lemma~\ref{lem:exponential:decay}  from the first output pulse
$\Delta_1$ to the initial input pulse $\Delta_0$.

\begin{lemma}\label{lem:initinput}
  There is a unique $\tilde{\Delta}_0$ such that every input pulse length
  $\Delta_0\geq\tilde{\Delta}_0$ guarantees $\Delta_1\geq \Delta$ as given in
  Lemma~\ref{lem:mid:pulse}.  Moreover,
  $\Delta_1 - \Delta \geq \big( 1 + \dup'(0) \big) \cdot (\Delta_0 -
  \tilde{\Delta}_0)$ for $\Delta_0 > \tilde{\Delta}_0$, provided
  $\Delta_0< \delta_\infty^\uparrow+ \eta^+$.
\end{lemma}
\begin{proof}
For the first pulse under the same worst-case adversarial choice as in Lemma~\ref{lem:mid:pulse},
the analogous considerations as in the proof of Lemma~\ref{lem:small:pulse} reveal
\begin{equation*}
\Delta_1 = \ddo( \Delta_{0} - \eta^+ - \delta_{\infty}^\uparrow)
+ \Delta_{0} -\eta^- - \eta^+ -
\delta_\infty^\uparrow 
\enspace.
\end{equation*}
Defining the auxiliary function
$g(\Delta_0) = \ddo( \Delta_0 - \eta^+ - \delta_{\infty}^\uparrow) + \Delta_0
-\eta^- - \eta^+ - \delta_\infty^\uparrow,$ it is apparent that
$\Delta_1=g(\Delta_0)$. Now, as
$\lim_{\Delta_0 \to \eta^+ + \delta_{\infty}^\uparrow - \delta_{\min}}
g(\Delta_0)\leq 0$ due to Lemma~\ref{lem:delta:min} and
$\lim_{\Delta_0 \to \eta^- + \eta^+ +
  \delta_{\infty}^\uparrow}g(\Delta_0)=\ddo(\eta^-)$, which is certainly (much)
larger than $\Delta$, cp.\ Lemma~\ref{lem:mid:pulse}, there is indeed a unique
$\tilde{\Delta}_0$ with $g(\tilde{\Delta}_0)=\Delta$ with the desired
properties.  The Lipschitz property is obtained exactly as in the proof of
Lemma~\ref{lem:exponential:decay}, by differentiating $g(\Delta_0)$ and using
$\Delta_0< \delta_\infty^\uparrow+ \eta^+$.
\end{proof}

We summarize the consequences of the previous lemmas in the following theorem,
which extends \cite[Thm.~12]{FNNS14:arxiv} to the \etamodel:
}

\begin{theorem}\label{thm:or:loop}
Consider the circuit in \figref{fig:circuit} subject to constraint 
\ifthenelse{\boolean{SHORTversion}}{$\eta^+ + \eta^- < \ddo(-\eta^+)-\deltamin$.}{\eqref{C}.}
The fed-back OR gate with a strictly causal \etachannel\  has the following output when the
input pulse has length~$\Delta_0$:
\begin{itemize}
\item If $\Delta_0 \geq \delta_\infty^\uparrow + \eta^+$, then the output has a
  single rising transition at time $0$.
\item If $\Delta_0 \leq \delta_\infty^\uparrow - \deltamin - \eta^+ - \eta^-$,
  then the output only contains the input pulse.
\item If
  $\delta_\infty^\uparrow - \deltamin - \eta^+ - \eta^- < \Delta_0 <
  \delta_\infty^\uparrow + \eta^+$, then the output may resolve to constant~$0$
  or $1$, or may be an (infinite) pulse train, with $\Delta_n\leq \Delta$
  \ifthenelse{\boolean{SHORTversion}}{for some $0<\Delta < \deltamin$}{} and
  duty cycle $\gamma_n\leq \gamma = \frac{\Delta}{\dup(-\Delta)+\eta^+} < 1$ for
  $n\geq 1$.  \ifthenelse{\boolean{SHORTversion}}{ 
    If $\Delta_n > \Delta$ for some $n$, bounded time later, the output resolves
    to $1$ and $\Delta_m > \Delta_{m-1}$ for all $m > n$.  }{ 
    If $\Delta_0 > \tilde{\Delta}_0$, the output resolves to $1$ within a
    stabilization time in the order of
    $\log_a (1/(\Delta_0 - \tilde{\Delta}_0))$ with $a=1+\dup'(0) > 1$.  }
\end{itemize}
\end{theorem}
\ifthenelse{\boolean{SHORTversion}}{}{
\begin{proof}
The statements of our theorem follow immediately from Lemmas~\ref{lem:big:pulse},
\ref{lem:mid:pulse}, and~\ref{lem:small:pulse}.
Lemma~\ref{lem:exponential:decay} in conjunction with Lemma~\ref{lem:initinput}
reveals that the number of 
generated pulses is in the order of  $\log_{a} (1/(\Delta_0 - \tilde{\Delta} ))$
with $a=1+\delta'(0)$.
\end{proof}
}

\ifthenelse{\boolean{SHORTversion}}{

Finally, a high-threshold buffer
with arbitrary threshold can be modeled by an exp-channel
with properly chosen~$V_{th}$:

}{

For dimensioning the high-threshold buffer, we can re-use Lemmas~13 and 14
from \cite{FNNS14:arxiv}:

\begin{lemma}[{\cite[Lem.~13]{FNNS14:arxiv}}]\label{lem:duty:cycle:delta}
Let~$C$ be an exp-channel with threshold~$V_{th}$ and initial value~$0$, and
let $0\leq \Gamma < V_{th}$.
Then there exists some~$\Theta>0$ such that every finite or infinite pulse train with pulse
lengths $\Theta_n \leq \Theta$, $n\geq0$, and duty cycles $\Gamma_n\leq \Gamma$, $n\geq1$, 
is mapped to the zero signal by~$C$.
\end{lemma}

}

\begin{lemma}[{\cite[Lem.~14]{FNNS14:arxiv}}]\label{lem:ht:exists}
  Let~$\Theta>0$ and $0\leq\Gamma<1$.  Then, there exists an exp-channel~$C$
  such that every finite or infinite pulse train with pulse lengths
  $\Theta_n\leq \Theta$, $n\geq0$, and duty cycles $\Gamma_n\leq \Gamma$,
  $n\geq1$, is mapped to the zero signal by~$C$.
\end{lemma}

By choosing $\Gamma = \gamma(1+\varepsilon)<1$ for some $\varepsilon>0$
sufficiently small and~$\Theta$ so large that the feed-back loop in
Figure~\ref{fig:circuit} has already locked to constant~$1$ at time $T+\Theta$,
where $T$ is the time when some pulse $\Delta_n$, $n\geq 1$, of the feed-back
loop with duty cycle $\gamma(1+\varepsilon)$ has started, we get the following:
If SPF input pulse lengths $\Delta_0$ and adversarial choices are such that no
$\Delta_n$ reaches duty cycle $\gamma(1+\varepsilon)$, the output of the
exp-channel is constant zero; otherwise, there is a single up-transition
(occurring only after $T+\Theta$) at the output.  Therefore:

\begin{theorem}
There is a circuit that solves unbounded SPF.
\end{theorem}

\ifthenelse{\boolean{SHORTversion}}{}{
\begin{proof}
  If $\Delta_0 < \delta_\infty^\uparrow - \deltamin - \eta^+ - \eta^-$,
  Theorem~\ref{thm:or:loop} ensures that the input of the high-threshold buffer
  is constant~0, and so is the output.  If
  $\Delta_0 > \delta_\infty^\uparrow + \eta^+$, then the input of the
  high-threshold buffer experiences a single up-transition (at time 0), and so
  does the output (eventually).

  For $\Delta_0$ in between, we distinguish two cases: (i) Suppose $\Delta_0$
  and the adversarial choices are such that no $\Delta_n$ ever reaches duty
  cycle $\gamma(1+\varepsilon)$. Then, the minimality of the period $P$ of the
  worst-case pulse train guaranteed by Lemma~\ref{lem:mid:pulse} implies that
  the input of the high-threshold buffer sees pulses with duration at most
  $\Theta$ and duty cycle at most $\Gamma$. Hence, Lemma~\ref{lem:ht:exists}
  guarantees a zero-output in this case.

  For the other case (ii), which is guaranteed to happen when
  $\Delta_0 > \tilde{\Delta}_0$ (but may also occur for smaller values of
  $\Delta_0$ in the case of certain adversarial choices), there is some time $T$
  where a 1-pulse $\Theta_n$ starts at the input of the exp-channel that will
  (along with its subsequent 0) have a duty cycle
  $\Gamma_n\geq\Gamma>\gamma$. Moreover, by time $T+\Theta$, the last input
  transition (to 1) has already occurred. Lemma~\ref{lem:ht:exists} not only
  guarantees that all pulses occurring before $T$ cancel, but also the ones that
  occur before time $T+\Theta$: after all, even a single, long pulse
  $\Theta_n=\Theta$ would still be canceled. Therefore, since the input of the
  exp-channel is already stable at 1 at time $T+\Theta$, only this final rising
  transition will eventually appear at the output.
\end{proof}
}

\section{Simulations}
\label{sec:simulations}

In this section, we complement the proof of faithfulness provided in
the previous section with simulation experiments and measurement
results, which confirm that our \etamodel\ indeed captures reality
better than the original involution model \cite{NFNS15:GLSVLSI}.
Whereas more experiments, with different technologies and more 
complex circuits (including multi-input gates), would be needed 
to actually claim improved model coverage, our results are nevertheless
encouraging.

\def\tikzmess{
\begin{tikzpicture}[circuit logic US, scale=0.8]

\draw[rounded corners=3pt, dotted, fill=black!20]
       (-0.4,-0.8) rectangle (6.7,0.4);
\node at (3.3,-1.1) {inverter chain};

\draw[rounded corners=3pt, dotted, fill=black!10]
       (0.2,0.5) rectangle (7.7,1.6);
\node [align=center] at (0,2.2) {on-chip sense\\amplifiers};

\node [not gate,small circuit symbols] (not1) at (0,0) {};
\node [not gate,small circuit symbols] (not2) at (1.2,0) {};
\node [not gate,small circuit symbols] (not3) at (2.4,0) {};
\node [not gate,small circuit symbols] (not4) at (3.6,0) {};
\node [not gate,small circuit symbols] (not5) at (4.8,0) {};
\node [not gate,small circuit symbols] (not6) at (6.0,0) {};
\node [not gate,small circuit symbols] (not7) at (7.5,0) {};

\node at (7.7,-0.5) {load};

\draw (not1.input)
       -- ++(left:0.4) node [yshift=-4pt,xshift=-4pt] {in};

\draw (not7.output)
       -- ++(right:0.3);

\path (not1.output)
        -- ++(right:0.2) 
        -- ++(up:1) node [buffer gate,small circuit symbols,rotate=90] (nnot1) {};
	;
\path (not2.output)
        -- ++(right:0.2) 
        -- ++(up:1) node [buffer gate,small circuit symbols,rotate=90] (nnot2) {};
	;
\path (not3.output)
        -- ++(right:0.2) 
        -- ++(up:1) node [buffer gate,small circuit symbols,rotate=90] (nnot3) {};
	;
\path (not4.output)
        -- ++(right:0.2) 
        -- ++(up:1) node [buffer gate,small circuit symbols,rotate=90] (nnot4) {};
	;
\path (not5.output)
        -- ++(right:0.2) 
        -- ++(up:1) node [buffer gate,small circuit symbols,rotate=90] (nnot5) {};
	;
\path (not6.output)
        -- ++(right:0.5) 
        -- ++(up:1) node [buffer gate,small circuit symbols,rotate=90] (nnot6) {};
	;

\foreach \x in {nnot1,nnot2,nnot3,nnot4,nnot5,nnot6} {
   \draw (\x.output) -- ++(up:0.3);
}
\node at (4,2) {to real-time oscilloscope};

\draw [name path=line1] (not1.output) -- (not2.input);
\path [name path=line2] (nnot1.input) -- ++(down:2);
\path [name intersections={of=line1 and line2,by=c1}];
\draw (nnot1.input) -- (c1) node[circle,inner sep=1pt,fill=black,draw] {};

\draw [name path=line1] (not2.output) -- (not3.input);
\path [name path=line2] (nnot2.input) -- ++(down:2);
\path [name intersections={of=line1 and line2,by=c2}];
\draw (nnot2.input) -- (c2) node[circle,inner sep=1pt,fill=black,draw] {};

\draw [name path=line1] (not3.output) -- (not4.input);
\path [name path=line2] (nnot3.input) -- ++(down:2);
\path [name intersections={of=line1 and line2,by=c3}];
\draw (nnot3.input) -- (c3) node[circle,inner sep=1pt,fill=black,draw] {};

\draw [name path=line1] (not4.output) -- (not5.input);
\path [name path=line2] (nnot4.input) -- ++(down:2);
\path [name intersections={of=line1 and line2,by=c4}];
\draw (nnot4.input) -- (c4) node[circle,inner sep=1pt,fill=black,draw] {};

\draw [name path=line1] (not5.output) -- (not6.input);
\path [name path=line2] (nnot5.input) -- ++(down:2);
\path [name intersections={of=line1 and line2,by=c5}];
\draw (nnot5.input) -- (c5) node[circle,inner sep=1pt,fill=black,draw] {};

\draw [name path=line1] (not6.output) -- (not7.input);
\path [name path=line2] (nnot6.input) -- ++(down:2);
\path [name intersections={of=line1 and line2,by=c6}];
\draw (nnot6.input) -- (c6) node[circle,inner sep=1pt,fill=black,draw] {};

\foreach \x/\y in {c1/1,c2/2,c3/3,c4/4,c5/5,c6/6} {
   \node[yshift=-9pt,xshift=-1pt] at (\x) {$Q_{\y}$};
}



\draw[->] (-0.1,1.7) -- ++(0.5,-0.3);

\end{tikzpicture}
}

\begin{figure}
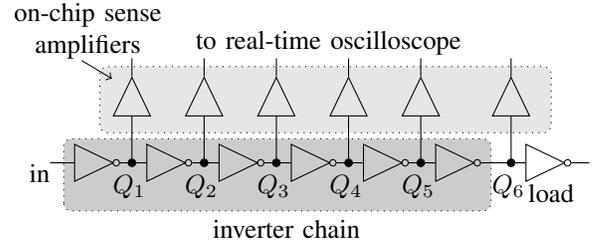

  \centerline{
    \tikzmess
  }
  \vspace{-0.2cm}
  \caption{\small\em Schematics of the ASIC used for validation measurements.
    It combines an inverter chain with analog high-speed sense amplifiers.}
  \label{fig:asic}
  \vspace{-0.5cm}
\end{figure}

We employ the same experimental setup as in \cite{NFNS15:GLSVLSI}, which uses
UMC-90\;nm and UMC-65\;nm bulk CMOS 7-stage inverter chains as the primary
targets. For UMC-65, we resorted to Spice simulations of a standard cell library
implementation, for UMC-90, we relied on a custom ASIC \cite{HSDZ12:ToNS}.  The
latter provides a 7-stage inverter chain built from 700\;nm x 80\;nm (W x L)
pMOS and 360\;nm x 80\;nm nMOS transistors, with threshold voltages 0.29\;V and
0.26\;V, respectively, and a nominal supply voltage of $V_{DD}=1$\;V.  As all
inverter outputs are connected to on-chip low-intrusive high-speed analog sense
amplifiers (gain 0.15, -3~dB cutoff frequency 8.5~GHz, input load equivalent to
3 inverter inputs), see Fig.~\ref{fig:asic}, which can directly drive the
50\;$\Omega$ input of a high-speed real-time oscilloscope, the ASIC facilitates
the faithful analog recording of all signal waveforms. Independent power
supplies and grounds for inverters and amplifiers also facilitate measurements
with different digital supply voltages $V_{DD}$.

For convenience, we provide the delay functions determined in
\cite{NFNS15:GLSVLSI} in Fig.~\ref{fig:delta90_all} ($\ddo$ for UMC-90,
measurements).

\begin{figure}[t]
  \centerline{
    \includegraphics[scale=0.8]{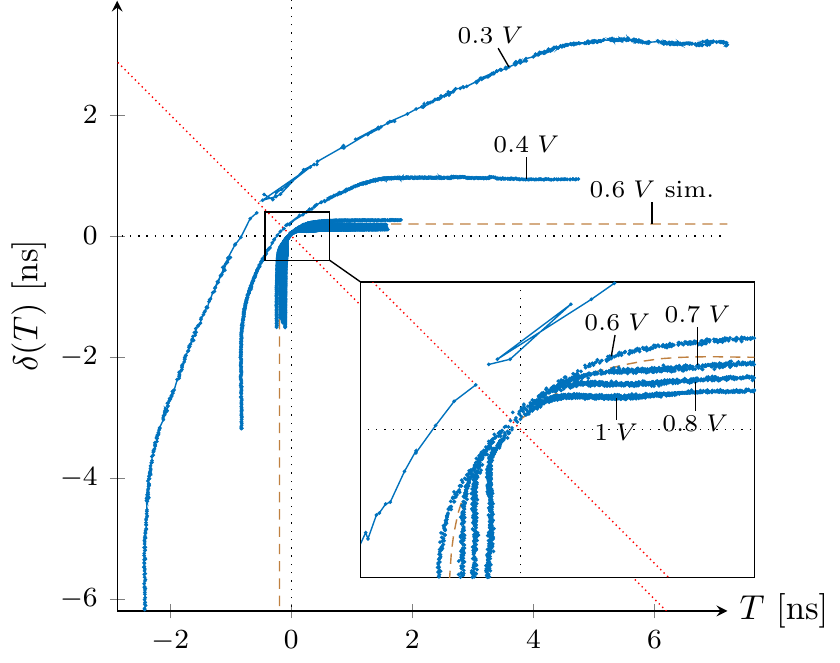}}
  \vspace{-0.1cm}
  \caption{\small\em Measured $\ddo$ for UMC-90 inverter chain for~$V_{DD} \in \{0.3, 0.4, 0.6, 0.7, 0.8, 1\}$\;V and simulated (dashed brown) $\ddo$ for $V_{DD} = 0.6$\;V, taken from  \cite[Fig.~7]{NFNS15:GLSVLSI}.
}
  \label{fig:delta90_all}
  \vspace{-0.5cm}
\end{figure}

In order to validate the \etamodel, we use the following general approach: Given
simulated/measured output waveforms of a single inverter excited by input pulses
of different width, we compare (i) the digital output obtained from the
simulated/measured waveforms with (ii) the predictions for some given delay
function.  The differences of the transition times of predicted and real digital
output is a measure of modeling inaccuracy of the original involution model.  If
these differences can be compensated by suitable output shifts within
$[\eta^-,\eta^+]$, however, we can claim that the \etamodel\ matches the real
behavior of the circuit for the given waveforms. Since faithfulness puts the
severe constraint $\eta^+ + \eta^- < \ddo(-\eta^+) - \deltamin$ on
$\eta^+,\eta^-$, \ifthenelse{\boolean{SHORTversion}}{}{recall
  Lemma~\ref{lem:mid:pulse},} it is not clear under which conditions this claim
indeed holds. In our evaluation, $\eta^+$ was first set to a suitable value
($\eta^+ > 0$) and afterwards $\eta^-$ was calculated according to
$\eta^-=\ddo(-\eta^+)-\deltamin-\eta^+$.  Clearly, this results in different
$\boeta$ bounds in each of the figures below.

The particular questions addressed in our experiments are the following: Is the
allowed range for~$\eta^+$ and~$\eta^-$ sufficient for the \etamodel\ to
capture: (a) The circuit behavior under variations of certain operating
conditions. After all, circuit delays change with varying supply voltage and
temperature, so the question remains to what extent the resulting fluctuations
are covered by the \etamodel.  (b) The circuit behavior under process
variations. In general, circuit delays vary from manufactured chip to chip, so
the question arises whether the \etamodel\ based on a ``typical'' delay function
covers typical process variations.  (c) The real behavior of our inverter chain
with a (suitably parametrized) standard involution function, in particular for
exp-channels.  This would simplify model calibration, as it is typically easier
to determine the exp-channel model parameters for a given
circuit~\cite{BDJCAVH00}, rather than its entire delay function.

\begin{figure*}
  \centering
  \begin{subfigure}{.32\linewidth}
    \begin{tikzpicture}[font=\footnotesize]
      \begin{axis}[
        width=0.95\linewidth,
        height=0.8\linewidth,
        xmin=-1e-11,
        xmax=1.3e-10,
        ymin=-6e-13,
        ymax=6e-13,
        xlabel={previous-output-to-input delay (T) [ps]},
        xtick={0,5e-11,1e-10},
        xticklabels={0,50,100},
        xtick scale label code/.code={},
        ylabel={deviation (D) [ps]},
        ytick={-4e-13,-2e-13,0,2e-13,4e-13},
        yticklabels={-0.4,-0.2,0,0.2,0.4},
        ytick scale label code/.code={},
        grid=major,
        grid style = {dashed, gray!50},
        legend entries={$\ddo$, $\dup$, $\eta$},
        legend cell align=right,
        legend pos = north east,
        legend style={at={(1,0.01)},anchor={south east},outer sep=1mm}
        ]
        \addplot [only marks, mark=*, color=Spectral-PRG-1]
        table[x , y]{inv_t4_noCL_vddVar1_EtaPlotDown.dat};
        \addplot [only marks, mark=diamond*, color=Spectral-PRG-3]
        table[x , y]{inv_t4_noCL_vddVar1_EtaPlotUp.dat};
        \addplot [color=Spectral-PRG-4, very thick] coordinates {(-1e-11,3e-13) (1.3e-10,3e-13)} ;
        \addplot [color=Spectral-PRG-4, very thick] coordinates {(-1e-11,-3.859649e-13) (1.3e-10,-3.859649e-13)} ;
      \end{axis}%
    \end{tikzpicture}%
    \caption{\small\em Power supply variations of $1$ \%.}
    \label{fig:vddVar}
  \end{subfigure}
  \begin{subfigure}{.32\linewidth}
    \begin{tikzpicture}[font=\footnotesize]
      \begin{axis}[
        width=0.95\linewidth,
        height=0.8\linewidth,
        xmin=-1e-11,
        xmax=1.3e-10,
        ymin=-1e-12,
        ymax=2e-13,
        xlabel={previous-output-to-input delay (T) [ps]},
        xtick={0,5e-11,1e-10},
        xticklabels={0,50,100},
        xtick scale label code/.code={},
        ylabel={deviation (D) [ps]},
        ytick scale label code/.code={},
        grid=major,
        grid style = {dashed, gray!50},
        legend entries={$\ddo$, $\dup$, $\eta$},
        legend cell align=right,
        legend pos = north east,
        legend style={at={(1,0.99)},anchor={north east},outer sep=1mm}
        ]
        \addplot [only marks, mark=*, color=Spectral-PRG-1]
        table[x , y]{highWidth10_EtaPlotDown.dat};
        \addplot [only marks, mark=diamond*, color=Spectral-PRG-3]
        table[x , y]{highWidth10_EtaPlotUp.dat};
        \addplot [color=Spectral-PRG-4, very thick] coordinates {(-1e-11,1e-13) (1.4e-10,1e-13)} ;
        \addplot [color=Spectral-PRG-4, very thick] coordinates {(-1e-11,-8.370927e-13) (1.4e-10,-8.370927e-13)} ;
      \end{axis}%
    \end{tikzpicture}%
    \caption{\small\em Transistor width increase of $10$ \%.}
    \label{fig:transWidthHigh}
  \end{subfigure}
  \begin{subfigure}{.32\linewidth}
    \begin{tikzpicture}[font=\footnotesize]
      \begin{axis}[
        width=0.95\linewidth,
        height=0.8\linewidth,
        xmin=-1e-11,
        xmax=1.3e-10,
        ymin=-1e-13,
        ymax=7e-13,
        xlabel={previous-output-to-input delay (T) [ps]},
        xtick={0,5e-11,1e-10},
        xticklabels={0,50,100},
        xtick scale label code/.code={},
        ylabel={deviation (D) [ps]},
        ytick={0,2e-13,4e-13,6e-13},
        yticklabels={0,0.2,0.4,0.6},
        ytick scale label code/.code={},
        grid=major,
        grid style = {dashed, gray!50},
        legend entries={$\ddo$, $\dup$, $\eta$},
        legend cell align=right,
        legend pos = north east,
        legend style={at={(1,0.1)},anchor={south east},outer sep=1mm}
        ]
        \addplot [only marks, mark=*, color=Spectral-PRG-1]
        table[x , y]{inv_t4_noCL_lowWidth10_EtaPlotDown.dat};
        \addplot [only marks, mark=diamond*, color=Spectral-PRG-3]
        table[x , y]{inv_t4_noCL_lowWidth10_EtaPlotUp.dat};
        \addplot [color=Spectral-PRG-4, very thick] coordinates {(-1e-11,4.777777e-13) (1.3e-10,4.777777e-13)} ;
        \addplot [color=Spectral-PRG-4, very thick] coordinates {(-1e-11,-4.511278e-14) (1.3e-10,-4.511278e-14)} ;
      \end{axis}%
    \end{tikzpicture}%
    \caption{\small\em Transistor width reduction of $10$ \%.}
    \label{fig:transWidthLow}
  \end{subfigure}
  \caption{\small\em Deviation between predicted and actual $V_{TH}$ crossings for
    different variations.}
  \vspace{-0.5cm}
\end{figure*}
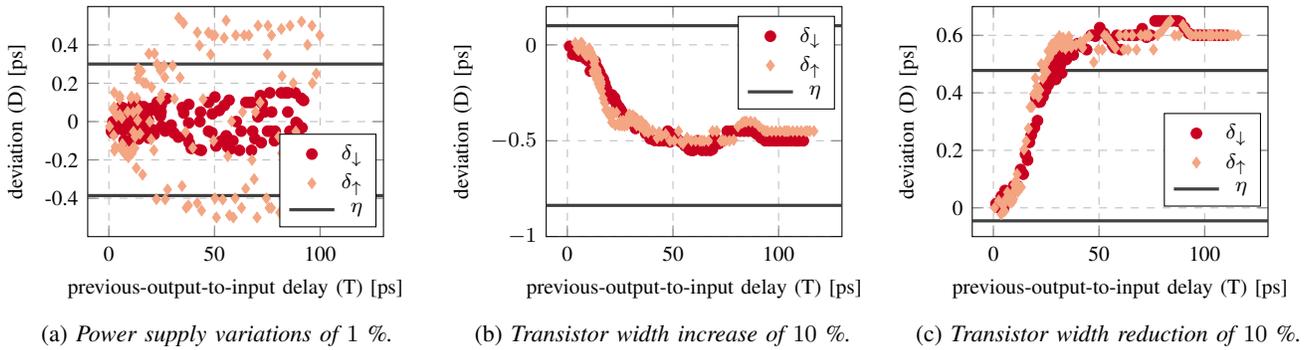

To investigate question~(a), i.e., the robustness against voltage variations, we
added a sine wave to the voltage supply source (nominally $1.2$ V $=V_{DD}$)
with a period similar to the full range switching time of the inverter and a
magnitude of $0.012$ V ($1$ \% of $V_{DD}$). We applied pulses with differing
width to the input of the inverter and recorded the output, whereat the phase of
the sine wave was set for each pulse randomly between $0$ and $360$ degrees. In
Fig.~\ref{fig:vddVar}, the deviation $D$ between the prediction and the actual
crossing over the previous-output-to-input delay $T$ is shown. Despite the
stringent bounds on $\boeta$, it is possible to fully cover the resulting delay
variations for low $T$, for higher values however, the \etamodel\ does no longer
apply. Please note that the huge difference between $\ddo$ and $\dup$ can be
easily explained by the fact that $\dup$ results in a falling transition at the
output of the inverter. For this transition, the transistor connecting the
output to the power supply gets closed more and more, reducing also the impact
of the voltage variations. (When varying the ground level, the reverse case can
be observed.)

To answer question~(b), we chose to vary the transistor width, which
increases/decreases the maximum current and allows us to model variations of
resistance and capacitance as well. The simulations themselves were carried out
in the same fashion as described in the last paragraph, except that
$V_{DD}=1.2$~V was constant. Fig.~\ref{fig:transWidthHigh} shows the results for
$10$ \% wider transistors, where the $\boeta$-bound is even bigger than
required. In contrast, the deviations for $10$ \% narrower ones
(Fig.~\ref{fig:transWidthLow}) exceed the $\boeta$-bound with increasing values
of $T$. Unlike $V_{DD}$ variations, varying transistor sizes, as expected,
either increases or decreases the delay.  This can be seen very clearly in the
figures, as one trace is well below and one well above $D=0$.

For question~(c), we tried to fit the parameters of the involution 
function (\ref{eq:exp}) for exp-channels w.r.t.\ the measurement data 
published in \cite{NFNS15:GLSVLSI} and evaluated the deviations $D$
between the resulting model predictions and the real digital 
output. Whereas the deviations over the whole range of $T$ exceed
the feasible $\boeta$-bounds, one can observe that even this very simple
exp-channel only results in minor mispredictions near $T=0$.
As shown in Fig.~\ref{fig:expChannelFit}, it again turns out that,
when using the resulting involution function, excessive deviations 
occur (quite naturally) for large values of $T$ only.

\begin{figure}
  \centering
  \begin{tikzpicture}[font=\footnotesize]
    \begin{axis}[
      width=0.8\linewidth,
      height=0.55\linewidth,
      xmin=-1e-10,
      xmax=1.7e-9,
      ymin=-0.82e-10,
      ymax=2e-11,
      xlabel={previous-output-to-input delay (T) [ns]},
      xtick scale label code/.code={},
      ylabel={deviation (D) [ps]},
      ytick={0,-2e-11,-4e-11,-6e-11,-8e-11},
      yticklabels={0,-20,-40,-60,-80},
      ytick scale label code/.code={},
      grid=major,
      grid style = {dashed, gray!50},
      legend entries={$\ddo$, $\dup$, $\eta$},
      legend cell align=right,
      legend pos = north east,
      legend style={at={(1,0.99)},anchor={north east},outer sep=1mm}
      ]
      \addplot [only marks, mark=*, mark size=2pt, color=Spectral-PRG-1]
      table[x , y]{V08EtaPlotDown.dat};
      \addplot [only marks, mark=diamond*, mark size=2pt, color=Spectral-PRG-3]
      table[x , y]{V08EtaPlotUp.dat};
      \addplot [color=Spectral-PRG-4, very thick] coordinates {(-1e-10,7.2e-12) (1.9e-9,7.2e-12)} ;
      \addplot [color=Spectral-PRG-4, very thick] coordinates {(-1e-10,-1.254888e-11) (1.9e-9,-1.254888e-11)} ;
    \end{axis}%
  \end{tikzpicture}%
  \vspace{-0.2cm}
  \caption{\small\em Fitting an exp-channel involution to measured data.}
  \label{fig:expChannelFit}
\vspace{-0.5cm}
\end{figure}
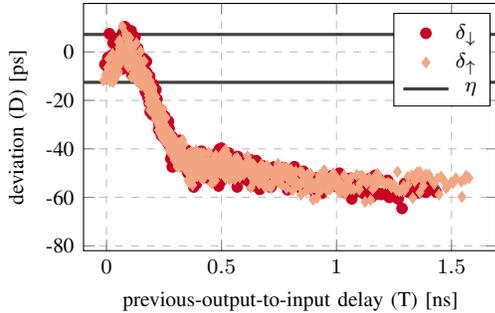

We hence conclude that the \etamodel\ indeed improves the modeling
accuracy of the original involution model, despite the fact that
the allowed non-determinism, i.e., $\boeta$, is quite restricted.
Moreover, our simulation experiments indicate that the absolute deviations $|D|$ between 
model predictions and real traces is increasing with increasing
previous-output-to-input delay $T$, making it possible to fully compensate $D$
via $\boeta$ near $T=0$. This is crucial, as our $\boeta$-bounds
result from proving faithfulness, which involves the range $T \in
[-\deltamin,0]$ only. For larger $T$, $D$ grows bigger, but in this
region, it might be feasible to also increase the allowed non-determinism
as these values are almost irrelevant w.r.t. faithfulness. 

\section{Conclusions and Future Work}
\label{sec:conclusions}

We proved the surprising fact that adding non-determinism
to the delays of involution channels, the only delay model known
so far that is faithful for the SPF problem, does not invalidate
faithfulness. As confirmed by some simulation experiments and
even measurements, noise, varying operating conditions and process
parameter variations hence do not a priori rule out faithful
continuous-time, binary value models. Part of our future work
will be devoted to further increase the level of non-determinism sustained
by our model, the handling of more complex circuits, and the first steps 
for incorporating the \etamodel\ in a suitable formal verification tool.

\end{document}